\documentclass[runningheads,12pt]{llncs} 

\usepackage[utf8]{inputenc}
\usepackage[english]{babel}

\usepackage{amsmath}
\usepackage{amssymb}
\usepackage[ruled]{algorithm2e}
\usepackage{lmodern}
\usepackage{microtype}
\usepackage{graphicx}
\usepackage{tikz}
\usetikzlibrary{shapes, calc}
\usepackage{hyperref}
\usepackage{comment}

\makeatletter
\tikzset{circle split part fill/.style  args={#1,#2}{%
 alias=tmp@name,
  postaction={
    insert path={
     \pgfextra{
     \pgfpointdiff{\pgfpointanchor{\pgf@node@name}{center}}%
                  {\pgfpointanchor{\pgf@node@name}{east}}%
     \pgfmathsetmacro\insiderad{\pgf@x}
      \fill[#1] (\pgf@node@name.base) ([xshift=-\pgflinewidth]\pgf@node@name.east) arc
                          (0:180:\insiderad-\pgflinewidth)--cycle;
      \fill[#2] (\pgf@node@name.base) ([xshift=\pgflinewidth]\pgf@node@name.west)  arc
                           (180:360:\insiderad-\pgflinewidth)--cycle;            
         }}}}}  
 \makeatother

\addto\extrasenglish{%
}

\newcommand{\one}{\mathbf{1}}
\newcommand{\ZZ}{\mathbb{Z}}
\newcommand{\RR}{\mathbb{R}}

\DeclareMathOperator{\outdegree}{outdeg}

\DeclareMathOperator{\supp}{supp}
\DeclareMathOperator{\rank}{rank}
\DeclareMathOperator{\script}{script}
\DeclareMathOperator{\dist}{dist}
\DeclareMathOperator{\dgon}{dgon}

\begin{document}

\title{Constructing Tree Decompositions of Graphs with Bounded Gonality%
\thanks{This research was initiated at the Sandpiles and Chip Firing Workshop, held November 25--26, 2019 at the Centre for Complex Systems Studies, Utrecht University.}}

\author{Hans L.~Bodlaender\inst{1}%
\and Josse van Dobben de Bruyn \inst{2}%
\thanks{Supported by NWO grant 613.009.127.}%
\and \\ Dion Gijswijt \inst{2} %
\and Harry Smit \inst{3}%
}

\institute{
Department of Information and Computing Sciences, Utrecht University, P.O.~Box 80.089, 3508 YB Utrecht, the Netherlands. \email{h.l.bodlaender@uu.nl}
\and
Delft Institute of Applied Mathematics, Delft University of Technology, the Netherlands.
\email{\{j.vandobbendebruyn,d.c.gijswijt\}@tudelft.nl}
\and
Department of Mathematics, Utrecht University, P.O.~Box 80.010, 3508 TA Utrecht, the Netherlands.
\email{h.j.smit@uu.nl}
}

\authorrunning{Bodlaender, van Dobben de Bruyn, Gijswijt, Smit}

\maketitle
\begin{abstract}
    In this paper, we give a constructive proof of the fact that the treewidth of a graph is at most its divisorial gonality. The proof gives a polynomial time algorithm to construct a tree decomposition of 
    width at most $k$, when an effective divisor of degree $k$ that reaches all vertices is given.
    We also give a similar result for two related notions: stable divisorial gonality and stable gonality.
\end{abstract}

\section{Introduction}
\label{section:introduction}
In this paper, we investigate the relation between well studied graph parameters: treewidth and divisorial gonality. In particular, we give a constructive proof that the treewidth of a graph is at most its divisorial gonality.

Treewidth is a graph parameter with a long history. Its first appearance was under the name of {\em dimension}, in 1972, by Bertele and Briochi~\cite{BerteleB72}. It was rediscovered several times since, under different names (see e.g.~\cite{Bodlaender98}). Robertson and Seymour introduced the notions of {\em treewidth} and
{\em tree decompositions} in their fundamental work on graph minors; these notions became the dominant terminology. 

The notion of divisorial gonality finds its origin in algebraic geometry. Baker and Norine~\cite{BakerN07} developed a divisor theory on graphs in analogy with divisor theory on curves, proving a Riemann--Roch theorem for graphs. The graph analog of gonality for curves was introduced by Baker~\cite{baker2008specialization}. To distinguish it from other notions of gonality (which we discuss briefly in Section~\ref{section:othernotions}), we denote the version we study by \emph{divisorial gonality}.
Divisorial gonality can be described in terms of a chip firing game. A placement of $k$ chips on the vertices of a graph (where vertices can have $0$ or more chips) is called an {\em effective divisor} of {\em degree $k$}. Under certain rules (see Section~\ref{section:preliminaries}), sets of vertices can {\em fire}, causing some of the chips to move to different vertices. The divisorial gonality of a graph is the minimum degree of an effective divisor such that for each vertex $v$, there is a firing sequence ending with a configuration with at least one chip at $v$.

A non-constructive proof that the treewidth is never larger than the divisorial gonality of a graph was
given by Van Dobben de Bruyn and Gijswijt \cite{JosseGijswijtTreewidth}. This proof was based on the characterization of
treewidth in terms of brambles, due to Seymour and Thomas~\cite{SeymourT93}. In this paper, we give
a constructive proof of the same fact. We formulate our proof in terms of a search game characterization
of treewidth, but with small modifications, we can also obtain a corresponding tree decomposition.
The proof also yields a polynomial time algorithm that, when given an effective divisor of degree $k$, constructs a search strategy with at most $k+1$ searchers and a tree decomposition of width at most $k$ of the input graph.

This paper is organized as follows. Some preliminaries are given in Section~\ref{section:preliminaries}. 
In Section~\ref{section:search}, we prove the main result with help of a characterization of treewidth in terms of a search game and discuss that we also can obtain a tree decomposition of width equal to the degree of a given effective divisor that reaches all vertices.
An example is given in Section~\ref{section:example}.
In Section~\ref{section:othernotions}, we give constructive proofs that bound the treewidth of a graph
in terms of two related other notions of gonality.

\section{Preliminaries}
\label{section:preliminaries}

\subsection{Graphs}
In this paper, all graphs are assumed to be finite. We allow multiple edges, but no loops. Let $G=(V,E)$ be a graph. For disjoint $U,W\subseteq V$ we denote by $E(U,W)$ the set of edges with one end in $U$ and one end in $W$, and use the shorthand $\delta(U)=E(U,V\setminus U)$. The \emph{degree} of a vertex $v\in V$ is $\deg(v)=|\delta(\{v\})|$, and given $v\in U\subseteq V$ we denote by $\outdegree_U(v)=|E(\{v\},V\setminus U)|$ the number of edges from $v$ to $V\setminus U$. By $N(U)$ we denote the set of vertices in $V\setminus U$ that have a neighbour in $U$. The \emph{Laplacian} of $G$ is the matrix $Q(G)\in \RR^{V\times V}$ given by 
$$
Q_{uv}=\begin{cases}
\deg(u)&\text{if $u=v$},\\
-|E(\{u\},\{v\})|&\text{otherwise}.
\end{cases}
$$

\subsection{Divisors and gonality}
Let $G=(V,E)$ be a connected graph with Laplacian matrix $Q=Q(G)$. A \emph{divisor} on $G$ is an integer vector $D\in \ZZ^V$. The \emph{degree} of $D$ is $\deg(D)=\sum_{v\in V}D(v)$. We say that a divisor $D$ is \emph{effective} if $D\geq 0$, i.e., $D(v)\geq 0$ for all $v\in V$. 

The divisorial gonality can be defined in a number of equivalent ways. Most intuitive is the definition in terms of a chip firing game. Before giving that definition, we first give the more formal definition,
which is needed in some of our proofs.

Two divisors $D$ and $D'$ are \emph{equivalent} (notation: $D\sim D'$) if $D'=D-Qx$ for some $x\in \ZZ^V$. Note that equivalent divisors have the same degree since $Q^T\one=0$. If $D$ and $D'$ are equivalent then, since the null space of $Q$ consists of all scalar multiples of $\one$, $D'=D-Qx$ has a unique solution $x\in \ZZ^V$ that is nonnegative and has $x_v=0$ for at least one vertex $v$. We denote this $x$ by $\script(D,D')$ and write $\dist(D,D')=\max\{x_v:v\in V\}$. Note that if $t=\dist(D,D')$, then $\script(D',D)=t\one-x$ and thus $\dist(D',D)=\dist(D,D')$. If $D,D',D''$ are pairwise equivalent, then we have the triangle inequality $\dist(D,D'')\leq \dist(D,D')+\dist(D',D'')$ as $\script(D,D'')=\script(D,D')+\script(D',D'')-c\one$ for some nonnegative integer $c$. 

Let $D$ be a divisor. If $D$ is equivalent to an effective divisor, then we define
\begin{alignat*}{1}
\rank(D)=
\max\{k\in \ZZ_{\geq 0}:\  &D-E \text{ is equivalent to an effective divisor}\\
&\text{for every effective divisor $E$ of degree at most $k$}\}.
\end{alignat*}
If $D$ is not equivalent to an effective divisor, we set $\rank(D)=-1$. The \emph{divisorial gonality} of a graph $G$ is defined as
$$
\dgon(G)=\min\{\deg(D) : \rank(D)\geq 1\}.
$$

In the remainder of the paper, we will only consider effective divisors. Given an effective divisor $D$, we can view $D$ as a chip configuration with $D(v)$ chips on vertex $v$.
If $U\subset V$ is such that $\outdegree_U(v)\leq D(v)$ for every $v\in U$ (i.e., each vertex has at least as many chips as it has edges to vertices outside $U$), then we say that \emph{$U$ can be fired}.
If this is the case, then \emph{firing $U$} means that every vertex in $U$ gives chips to each of its neighbours outside $U$, one chip for every edge connecting to that neighbour.
The resulting chip configuration is the divisor $D' = D - Q\one_U$.
The assumption $\outdegree_U(v) \leq D(v)$ guarantees that the number of chips on each vertex remains nonnegative, i.e., that $D'$ is effective.

If we can go from $D$ to $D'$ by sequentially firing a number of subsets, then clearly $D\sim D'$. The converse is also true (part (i) of the next lemma) as was shown in \cite[Lemma 1.3]{JosseGijswijtTreewidth}. 
\begin{lemma}\label{Lem:chain}
Let $D$ and $D'$ be equivalent effective divisors. 
\begin{itemize}
    \item[\textup{(i)}] There is a unique increasing chain $\emptyset\subsetneq U_1\subseteq U_2\subseteq\cdots\subseteq U_t\subsetneq V$ of subsets on which we can fire in sequence to obtain $D'$ from $D$. That is, setting $D_0=D$ and $D_i=D_{i-1}-Q\one_{U_i}$ for $i=1,\ldots, t$ we have $D_t=D'$ and $D_i$ is effective for all $i=0,\ldots, t$.
    \item[\textup{(ii)}] We have $t=\dist(D,D')\leq \deg(D)\cdot|V|$.
\end{itemize}
\end{lemma}
\begin{proof}
Let $x=\script(D,D')$ and let $t=\dist(D,D')=\max\{x(v):v\in V\}$. We let $U_1\subseteq\cdots\subseteq U_t$ be the level set decomposition of $x$. That is, 
$$U_i=\{v\in V: x(v)\geq t+1-i\}\qquad (i=1,\ldots, t).
$$
So $x=\sum_{i=1}^t \one_{U_i}$. To conclude the proof of part (i), it suffices to show that the divisors $D_0,D_1,\ldots, D_t$ are indeed effective. By assumption, this is true for $D_0=D$ and $D_t=D'$. Consider any $v\in V$. If $v\not\in U_t$, then $0\leq D_0(v)\leq D_1(v)\leq\cdots\leq D_t(v)$ since chips can only be added to $v$ when firing a subset not containing $v$.  Otherwise, let $i$ be the smallest index for which $v\in U_i$. Then 
$$0\leq D_0(v)\leq D_1(v)\leq \cdots\leq  D_{i-1}(v)$$
and 
$$D_{i-1}(v)\geq D_i(v)\geq D_{i+1}(v)\geq\cdots\geq D_t(v)\geq 0.$$ 
Hence, $D_i(v)\geq 0$ for all $i=0,\ldots, t$. 

For part (ii), we note that a set $U$ can occur at most $\deg(D)$ times in the chain $U_1\subseteq\cdots\subseteq U_t$ since each time we fire the set $U$ at least one chip leaves $U$. It follows that $t\leq \deg(D)\cdot |V|$.
\qed
\end{proof}


We see that the divisorial gonality of a graph $G$ is the minimum number $k$ such that there is
a starting configuration (divisor) with $k$ chips, such that for each vertex $x\in V$ there
is a sequence of sets we can fire such that $x$ receives a chip. Lemma~\ref{Lem:chain} shows that we
even can require these sets to be increasing.

For a given vertex $q$, a divisor $D\geq 0$ is called \emph{$q$-reduced} if there is no nonempty set $U\subseteq V\setminus\{q\}$ such that $D-Q\one_U\geq 0$. 
\begin{lemma}[{\cite[Proposition 3.1]{BakerN07}}]
Let $D$ be an effective divisor and let $q$ be a vertex. There is a unique $q$-reduced divisor equivalent to $D$. 
\end{lemma}
Let $D$ be an effective divisor and let $D_q$ be the $q$-reduced divisor equivalent to $D$. Suppose that $D\neq D_q$. By Lemma~\ref{Lem:chain} we obtain $D_q$ from $D$ by firing on a chain of sets $U_1\subseteq\cdots\subseteq U_t$ and, conversely, we obtain $D$ from $D_q$ by firing on the complements of $U_t,\ldots, U_1$. Since $D_q$ is $q$-reduced, it follows that $q$ is in the complement of $U_t$, and hence $q\not\in U_1$. It follows that $x=\script(D,D_q)$ satisfies $x_q=0$ and $D_q(q)\geq D(q)$. In particular, a divisor $D$ has positive rank if and only if for every $q \in V$ the $q$-reduced divisor equivalent to $D$ has at least one chip on vertex $q$.

Given an effective divisor $D$ and a vertex $q$, Dhar's algorithm~\cite{dhar1990self} finds in polynomial time a nonempty subset $U\subseteq V\setminus\{q\}$ on which we can fire, or concludes that $D$ is $q$-reduced. 
\begin{center}
\begin{algorithm}[ht]
\SetAlgoNoLine
\SetKwInOut{KwIn}{Input}
\SetKwInOut{KwOut}{Output}
\KwIn{Divisor $D\geq 0$ on $G$ and vertex $q$.}
\KwOut{Nonempty subset $U\subseteq V(G)\setminus\{q\}$ s.t. $D-Q\one_U\geq 0$ or $U=\emptyset$ if none exists.}

$U \leftarrow V\setminus\{q\}$\;
\While{$\outdegree_U(v)>D(v)$ for some $v\in U$}
{$U\leftarrow U\setminus\{v\}$}
\KwRet{U}
\caption{Dhar's burning algorithm}
\label{alg:Dhar}
\end{algorithm}
\end{center}
\begin{lemma}\label{lem:Dhar}
Dhar's algorithm is correct, and the output is the unique inclusionwise maximal subset $U\subseteq V\setminus\{q\}$ that can be fired.
\end{lemma}
\begin{proof}
The set returned by \autoref{alg:Dhar} can be fired, as it satisfies the requirement $\outdegree_U(v)\leq D(v)$ for every $v\in U$. To complete the proof it therefore suffices to show that $U$ contains every subset $W\subseteq V\setminus\{q\}$ that can be fired. 

Let $W\subseteq V\setminus\{q\}$ be any such subset. At the start of the algorithm $U=V\setminus\{q\}$ contains $W$. While $U\supseteq W$ we have $\outdegree_U(v)\leq \outdegree_W(v)\leq D(v)$ for any $v\in W$, so the algorithm never removes a vertex $v\in W$ from $U$.
\qed
\end{proof}

Note: in particular, \autoref{lem:Dhar} shows that the output of \autoref{alg:Dhar} does not depend on the order in which vertices are selected for removal.

If throughout the algorithm we keep for every vertex $v$ the number $\outdegree_U(v)$ and a list of vertices for which $\outdegree_U(v)>D(v)$, then we need only $O(|E|)$ updates, and we can implement the algorithm to run in time $O(|E|)$.

\begin{lemma}\label{Lem:distdecrease}
Let $D$ be an effective divisor on the graph $G=(V,E)$, let $q\in V$, and let $D_q$ be the $q$-reduced divisor equivalent to $D$. Let $U$ be the set returned by Dhar's algorithm when applied to $D$ and $q$, and suppose that $U\neq \emptyset$. Let $D'=D-Q\one_U$. Then $\dist(D',D_q)=\dist(D,D_q)-1$. 
\end{lemma}
\begin{proof}
 Let $x=\script(D,D_q)$. Since $D_q$ is $q$-reduced, we have $x_q=0$. On the other hand, since $D\neq D_q$ (as we can fire on $U$), the number $t=\max\{x_v:v\in V\}$ is positive. Let $W=\{v\in V:x_v=t\}$. By \autoref{Lem:chain}, we can fire on $W$, so by \autoref{lem:Dhar} we have $W\subseteq U$. 
 
 Let $x'=\script(D',D_q)$ and let $t'=\max\{x'_v:v\in V\}$. As $D_q$ is $q$-reduced, we have $x'_q=0$. Since there is a unique nonnegative $y\in \ZZ^V$ with $y_q=0$ and $D_q=D-Qy$, and we have $D-Qx=D_q=(D-Q\one_U)-Qx'$, it follows that $x=x'+\one_U$. Since $U\supseteq W$, it follows that $x-\one_W\geq x'$, and hence $t-1\geq t'$.
We find that $\dist(D',D_q)\leq \dist(D,D_q)-1$. Since $\dist(D,D')=1$, equality follows by the triangle inequality.
\qed
\end{proof}
Since $\dist(D,D_q)\leq \deg(D)\cdot |V(G)|$, we can find a $q$-reduced divisor equivalent to $D$ using no more than  $\deg(D)\cdot |V|$ applications of Dhar's algorithm.

\subsection{Treewidth and tree decompositions}

The notions of treewidth and tree decomposition were introduced by Robertson and Seymour~\cite{RobertsonS2}
in their fundamental work on graph minors.

Let $G=(V,E)$ be a graph, let $T=(I,F)$ be a tree, and let $X_i\subseteq V$ be a set of vertices 
(called {\em bag})
associated to $i$ for every node $i\in I$. The pair $(T,(X_i)_{i\in I})$ is a \emph{tree decomposition} of $G$ if it satisfies the following conditions: 
\begin{enumerate}
    \item $\bigcup_{i\in I} X_i = V$;
    \item for all $e=vw\in E$, there is an $i\in I$ with $v,w\in X_i$;
    \item for all $v\in V$, the set of nodes $I_v = \{i\in I ~|~ v\in X_i\}$ is connected (it induces a subtree of $T$).
\end{enumerate}
The \emph{width} of the tree decomposition is $\max_{i\in I} |X_i|-1$. The \emph{treewidth} of a $G$ is the minimum width of a tree decomposition of $G$. Note that the treewidth of a multigraph is equal to the treewidth of the underlying simple graph.

There are several notions that are equivalent to treewidth. We will use a notion that is based on a
Cops and Robbers game, introduced by Seymour and Thomas~\cite{SeymourT93}. Here, a number of searchers need to catch a fugitive. Searchers can move from a vertex in the graph to a `helicopter', or from a helicopter to any vertex in the graph. Between moves of searchers, the fugitive can move with infinite speed in the graph, but can only use paths that do not contain or go to a vertex with a searcher. The fugitive is captured when a searcher moves to the vertex with the fugitive, and the fugitive has no other vertex without a searcher he can move to. The location of the fugitive is 
known to the searchers at all times. We say that $k$ searchers can capture a fugitive in a graph $G$, if there is a strategy for $k$ searchers on $G$ that guarantees that the fugitive is captured. In the initial configuration, the fugitive can choose a vertex, and all searchers are in a helicopter.
A search strategy is \emph{monotone} if it is never possible for the fugitive to move to a vertex that had been unreachable before. In particular, in a monotone search strategy, there is never a path without searchers from the location of the fugitive to a vertex previously occupied by a searcher.

\begin{theorem}[Seymour and Thomas~\cite{SeymourT93}]
Let $G$ be a graph and $k$ a positive integer. The following statements are equivalent.
\begin{enumerate}
    \item The treewidth of $G$ is at most $k$.
    \item $k+1$ searchers can capture a fugitive in $G$.
    \item $k+1$ searchers can capture a fugitive in $G$ with a monotone search strategy.
\end{enumerate}
\end{theorem}

\section{Construction of a search strategy}
\label{section:search}


In this section, we present a polynomial time algorithm that,
given an effective divisor $D$ of degree $k$ as input,
constructs a monotone search strategy with $k+1$ searchers to capture the fugitive.

We start by providing a way to encode monotone search strategies. Let $G$ be a graph. For $X\subseteq V(G)$, the vertex set of a component of $G-X$ is called an \emph{$X$-flap}. A \emph{position} is a pair $(X,R)$, where $X\subseteq V(G)$ and $R$ is a union\footnote{Here we deviate from the definition of position at stated in~\cite{SeymourT93} in that we allow $R$ to consist of zero $X$-flaps or more than one $X$-flap.} of $X$-flaps (we allow $R=\emptyset$). The set $X$ represents the vertices occupied by searchers, and the fugitive can move freely within some $X$-flap contained in $R$ (if $R=\emptyset$, then the fugitive has been captured). In a monotone search strategy, the fugitive will remain confined to $R$, so placing searchers on vertices other than $R$ is of no use. Therefore, it suffices to consider three types of moves for the searchers: (a) remove searchers that are not necessary to confine the fugitive to $R$; (b) add searchers to $R$; (c) if $R$ consists of more than one $X$-flap, restrict attention to the $X$-flap $R_i\subset R$ containing the fugitive. This leads us to the following definition.    
 
\begin{definition}\label{def:monotone_strategy}
Let $G$ be a graph and let $k$ be a positive integer. A \emph{monotone search strategy} (MSS) with $k$ searchers for $G$ is a directed tree $T=(\mathcal{P},F)$ where $\mathcal{P}$ is a set of positions with $|X|\leq k$ for every $(X,R)\in \mathcal{P}$, and the following hold:
\begin{itemize}
\item[\textup{(i)}] The root of $T$ is $(\emptyset,V)$.
\item[\textup{(ii)}] If $(X,R)$ is a leaf of $T$, then $R=\emptyset$.
\item[\textup{(iii)}] Let $(X,R)$ be a non-leaf of $T$. Then $R\neq\emptyset$ and there is a set $X'\subseteq X\cup R$ such that exactly one of the following applies: 
\begin{itemize}
\item[\textup{(a)}] $X'\subset X$ and position $(X',R)$ is the unique out-neighbour of $(X,R)$.
\item[\textup{(b)}] $X'\supset X$ and position $(X',R')$ is the unique out-neighbour of $(X,R)$, where $R'=R\setminus X'$.
\item[\textup{(c)}] $X'=X$ and the out-neighbours of $(X,R)$ are the positions\\ $(X,R_1),\ldots, (X,R_t)$ where $t\geq 2$ and $R_1,\ldots, R_t$ are the $X$-flaps contained in $R$.
\end{itemize}   
\end{itemize}
If condition \textup{(ii)} does not necessarily hold, we say that $T$ is a \emph{partial} MSS. Note that we do not consider the root node to be a leaf even if it has degree $1$.
\end{definition}
It is clear that if $T$ is an MSS for $k$ searchers then, as the name suggests, $k$ searchers can capture the fugitive, the fugitive can never reach a vertex that it could not reach before, and a searcher is never placed on a vertex from which a searcher was previously removed. 

\begin{lemma}\label{lem:StrategySize}
Let $G$ be a graph on $n$ vertices and let $T$ be a (partial) MSS with $k$ searchers for $G$. Then $T$ has no more than $n^2+1$ nodes.
\end{lemma}
\begin{proof}
For any position $(X,R)$, define $f(X,R)=|R|(|X|+|R|)$. For any leaf node $(X,R)$ we have $f(X,R)\geq 0$. For any non-leaf node $(X,R)$, the value $f(X,R)$ is at least the sum of the values of its children plus the number of children. Indeed, in case (a) and (b) we have $f(X,R)\geq f(X',R')+1$, and in case (c) we have $f(X,R)\geq f(X,R_1)+\cdots+f(X,R_k)+k$ as can be easily verified. It follows that $f(X,R)$ is an upper bound on the number of descendants of $(X,R)$ in $T$. Since every non-root node is a descendant of the root, it follows that the total number of nodes is at most $1+f(\emptyset,V)=1+n^2$. 
\qed
\end{proof}

In the construction of an MSS we will use the following lemma.
\begin{lemma}\label{Lem:goodfiringset}
Let $R$ be an $X$-flap. Let $D$ be a positive rank effective divisor such that $X\subseteq \supp(D)$ and $R\cap\supp(D)=\emptyset$. Then we can find in polynomial time an effective divisor $D'\sim D$ such that $X\subseteq\supp(D')$, $R\cap\supp(D')=\emptyset$, and such that from $D'$ we can fire a subset $U$ with $U\cap R=\emptyset$ and $U\cap X\not=\emptyset$.
\end{lemma}
\begin{proof}
Let $q\in R$. Let $U$ be the set found by Dhar's algorithm. Since $R$ is connected and $U$ does not contain $R$, it follows that $U\cap R=\emptyset$ (otherwise $\outdegree_U(r)\geq 1>D(r)$ for some $r\in U\cap R$). If $U\cap X$ is nonempty, we set $D'=D$ and we are done. Otherwise, we set $D\leftarrow D-\one_U$. Then $X\subseteq \supp(D)$, $R\cap\supp(D)=\emptyset$ and we iterate. We must finish in no more than  $\deg(D)\cdot|V|$ iterations by Lemma~\ref{Lem:chain} and Lemma~\ref{Lem:distdecrease}. Hence, we can find the required $D'$ and $U$ in time $|E(G)|\cdot |V(G)|\deg(D)$.
\qed
\end{proof}

\paragraph{Construction of a monotone search strategy}
Let $G$ be a connected graph and let $D$ be an effective divisor on $G$ of positive rank. Let $k=\deg(D)$. We will construct an MSS for $k+1$ searchers on $G$. We do this by keeping a partial MSS, starting with only the root node $(\emptyset, V)$ and an edge to the node $(X,V\setminus X)$, where $X=\supp(D)$. Then, we iteratively grow $T$ at the leaves $(X,R)$ with $R\neq \emptyset$ until $T$ is an MSS. At each step, we also keep, for every leaf $(X,R)$ of $T$, an effective divisor $D'\sim D$ such that $X\subseteq \supp(D')$ and $R\cap \supp(D')=\emptyset$. We now describe the iterative procedure. 

While $T$ has a leaf $(X,R)$ with $R\neq\emptyset$, let $D'$ be the divisor associated to $(X,R)$ and perform one of the following steps. 
\begin{itemize}
\item[I.] If $R$ consists of multiple $X$-flaps $R_1,\ldots, R_t$, then we add nodes\\ $(X,R_1),\ldots, (X,R_t)$ as children of $(X,R)$ and associate $D'$ to each. Iterate. 
\item[II.] If $X'=N(R)$ is a strict subset of $X$, then add the node $(X',R)$ as a child of $(X,R)$, associate $D'$ to this node and iterate.
\item[III.] The remaining case is that $N(R)=X$ and $R$ is a single $X$-flap. By Lemma~\ref{Lem:goodfiringset} we can find an effective divisor $D''\sim D'$ such that $X\subseteq \supp(D'')$, $R\cap \supp(D'')=\emptyset$ and from $D''$ we can fire on a set $U$ such that $U\cap R=\emptyset$ and $U\cap X\neq\emptyset$. We set $U\cap X=\{s_1,s_2,\ldots, s_t\}$. That we can fire on $U$ implies that 
\begin{equation}\label{eq:enoughchips}
D''(s_i)\geq |N(s_i)\cap R|\quad\text{for $i=1,\ldots, t$}.
\end{equation}
For $i=1,\ldots, t$ we define positions $(X_i,R_i)$ and $(X'_i,R_i)$ as follows:
$$
X_i=X'_{i-1}\cup (N(s_i)\cap R),\quad R_i=R\setminus X_i,\qquad\text{and}\qquad X'_i=X_i\setminus\{s_i\},
$$
where we set $X'_0=X$. Using (\ref{eq:enoughchips}) and the fact that $X'_0\subseteq \supp(D'')$, it is easy to check that $|X_i'|\leq k$ and $|X_i|\leq k+1$ for every $i$. Since every edge in $\delta(R)$ has at least one endpoint in every $X'_i$, it follows that indeed $R_i$ is a union of $X'_i$-flaps (and of $X_i$-flaps). We add the path $(X,R)\to (X_1,R_1)\to (X'_1,R_1)\to\cdots\to (X'_t,R_t)$ to $T$ (it may happen that $(X_i,R_i)=(X_{i-1}',R_{i-1})$ in which case we leave out one of the two). We associate $D''-Q\one_U$ to the leaf $(X'_t,R_t)$.
\end{itemize}

By Lemma~\ref{lem:StrategySize}, we are done in at most $|V(G)|^2$ steps. This completes the construction. 
By combining the construction described above with that of the lemma below,
we obtain Theorem~\ref{theorem:polyconstruct}. Note that so far only a non-constructive
proof that the divisorial
gonality of a graph is an upper bound for the treewidth was known~\cite{JosseGijswijtTreewidth}.

\begin{lemma}
Let $T'=(\mathcal{P},F)$ be a monotone search strategy for $k$ searchers in the connected graph $G$ and let $T$ be the undirected tree obtained by ignoring the orientation of edges in $T'$. Then $(T, \{X\}_{(X,R)\in \mathcal{P}})$ is a tree decomposition of $G$ of width at most $k-1$. 
\label{lemma:search2td}
\end{lemma}   
\begin{proof}
It is clear that $V=\bigcup \{X:\exists (X,R)\in \mathcal{P}\}$ since a fugitive stationary at any given vertex can be captured. 

Let $v\in V$. We must show that the set of nodes $\{(X,R)\in \mathcal{P}:v\in X\}$ is a subtree of $T$. Equivalently, we must show that if node $(X_2,R_2)$ lies on a path from $(X_1,R_1)$ to $(X_3,R_3)$ in $T$, then $X_1\cap X_3\subseteq X_2$. It suffices to check this in two cases: the case that $(X_3,R_3)$ is a descendant of $(X_1,R_1)$ in $T'$, and the case that $(X_2,R_2)$ is the last common ancestor of $(X_1,R_1)$ and $(X_3,R_3)$. In the first case, it is easy to see that $X_3\subset X_2\cup R_2$ and $R_2\subseteq R_1$ hold. It follows that
$$
X_1\cap X_3\subseteq X_1\cap (X_2\cup R_2)\subseteq X_1\cap (X_2\cup R_1)\subseteq X_2
$$ 
since $X_1$ and $R_1$ are disjoint. In the second case, node $(X_2,R_2)$ has more than one out-neighbour, so its out-neighbours are positions $(X_2,R)$, where $R$ runs over the $X_2$-flaps contained in $R_2$. It follows that $X_1\subseteq X_2\cup R'$ and $X_3\subseteq X_2\cup R''$ for distinct $X_2$-flaps $R'$ and $R''$. Hence, $X_1\cap X_3\subset X_2$.  

To complete the proof, it suffices to show that for every edge $\{u,v\}\in E(G)$ there is some node $(X,R)$ of $T$ with $u,v\in X$. Suppose for contradiction that this is not the case for edge $\{u,v\}$. 

We first show that there is a node $(X,R)$ such that $u\in X$ and $v\in R$ (or vice versa). To this end, consider the nodes $(X,R)$ of $T$ with $u,v\in R$ (e.g. the root node), and take such a node that has maximum distance from the root. This node cannot be a leaf since $R\ni v$ is non-empty. Since $u$ and $v$ belong to the same $X$-flap, it follows by the maximality assumption that $(X,R)$ has a child $(X',R')$ with $u\in X'$ and $v\in R'$ (or vice versa). 

Now consider all nodes $(X,R)$ with $u\in X$ and $v\in R$ and take such a node for which the distance to the root is maximised. This node cannot be a leaf (since $R\ni v$ is non-empty). Consider a child $(X',R')$ of $(X,R)$. If we are in case (iii)(a) then $v\in R'$ and we must have $u\in X'$ since otherwise $R'$ is not a union of  $X'$-flaps as $\{u,v\}$ is an edge. This contradicts the maximality assumption. If we are in case (iii)(b), then $u\in X'$ and $v\in R'$ contradicting the maximality assumption. If we are in case (iii)(c), we may assume that $R'$ is the $X$-flap containing $v$ and again this contradicts the maximality assumption. 
\qed
\end{proof}

\begin{theorem}
There is a polynomial time algorithm that, when given a graph $G$ and an effective divisor of degree $k$, finds a tree decomposition of $G$ of width at most $k$.
\label{theorem:polyconstruct}
\end{theorem}

\section{An example}
\label{section:example}
We apply the constructions of the previous section to a relatively small example. Let $G$ be the graph as in Figure~\ref{fig:example_graph_G}. Let $D$ be the divisor on $G$ that has value $3$ on vertex $a$ and value $0$ elsewhere. 

\begin{figure}[htb]
\begin{center}
    \includegraphics[scale=1.0]{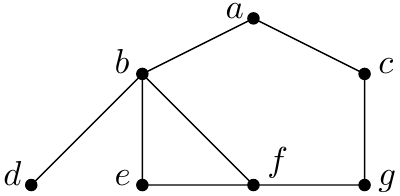} 
\end{center}

\caption{An example graph $G$. It has divisorial gonality equal to $3$.}
\label{fig:example_graph_G}
\end{figure}

If we follow the construction of Section~\ref{section:search}, we will end up with the monotone search strategy found in Figure~\ref{fig:monotone_search_on_G}. We start with the root node $(X,R)$ with $X=\emptyset$ and $R=V$ and connect it to the node $(\supp(D), V\setminus\supp(D))$. The three ways of growing the tree (steps I, II, III) are indicated in the picture. The four occurrences of step III are explained below.

\begin{figure}[htb]
\newcommand{\nodetekst}[2]{\rule{0pt}{0.8em}#1\nodepart{two}\rule{0pt}{0.8em}#2}

 \centering
 \begin{tikzpicture}
 \tikzset{mynode/.style={font=\small,inner sep=3pt, rounded corners,rectangle split, rectangle split parts=2,draw=black,thick, minimum height=1cm}}
 \node[font=\small,inner sep=3pt,rectangle split, rectangle split parts=2, minimum height=1cm, line width=0pt] at (-1,0) {\nodetekst{$X$}{$R$}};
 \node[mynode, very thick] (R) at (0,0) {\nodetekst{}{$abcdefg$}};
 \node[mynode] (R2) at (1.5,0) {\nodetekst{$a$}{$bcdefg$}};
 \node[mynode] (A) at (3.5,0) {\nodetekst{$abc$}{$defg$}};
 \node[mynode] (A2) at (4.6,0) {\nodetekst{$bc$}{$defg$}};
 \node[mynode] (B) at (6.3,0) {\nodetekst{$bc$}{$d$}};
 \node[mynode] (C) at (7.8,0) {\nodetekst{$b$}{$d$}};
 \node[mynode] (D) at (9.3,0) {\nodetekst{$bd$}{}};
 \node[mynode] (D2) at (10,0) {\nodetekst{$d$}{}};
 
 \node[mynode] (E) at (2,-2) {\nodetekst{$bc$}{$efg$}};
 \node[mynode] (F) at (3.7,-2) {\nodetekst{$bcg$}{$ef$}};
 \node[mynode] (F2) at (4.6,-2) {\nodetekst{$bg$}{$ef$}};
  
 \node[mynode] (G) at (6.4,-2) {\nodetekst{$befg$}{}};
 \node[mynode] (G2) at (7.5,-2) {\nodetekst{$efg$}{}};
 \node[mynode] (G3) at (8.5,-2) {\nodetekst{$ef$}{}};
   
 \draw[thick] (R)--(R2);
 \draw[thick] (A)--(A2);
 \draw[thick] (D)--(D2);
 \draw[thick] (F)--(F2);
 \draw[thick] (G)--(G2)--(G3);

 \draw[thick, ->, shorten >=3pt] (R2)--node [above] {\footnotesize III}node [below] {\footnotesize (1)}(A);
 \draw[thick, ->, shorten >=3pt] (A2)--node [above] {\footnotesize I}(B);
 \draw[thick, ->, shorten >=3pt] (B)--node [above] {\footnotesize II}(C);
 \draw[thick, ->, shorten >=3pt] (C)--node [above] {\footnotesize III}node [below] {\footnotesize (2)}(D);
 \draw[thick, ->, shorten >=3pt] (E)--node [above] {\footnotesize III}node [below] {\footnotesize (3)}(F);
 \draw[thick, ->, shorten >=3pt] (F2)--node [above] {\footnotesize III}node [below] {\footnotesize (4)}(G);

 \draw[thick, rounded corners, ->, shorten >=3pt] (A2)--node[right]{\footnotesize I}(4.6,-1)--(1,-1)--node[left]{\footnotesize I}(1,-2)--(E);
 \end{tikzpicture} 

\caption{The monotone search strategy obtained from $G$ with divisor $D=3a$. Each node shows the corresponding pair $(X,R)$ with the root being $(\emptyset, \{a,b,c,d,e,f,g\})$. The labels I--III refer to the steps in the construction.}
\label{fig:monotone_search_on_G}
\end{figure}
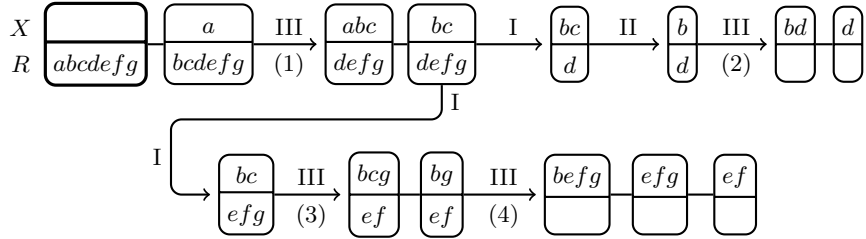

For compactness of notation, we write the divisors as a formal sum. For instance, if $D'$ has $2$ chips on $b$ and $1$ chip on $g$, we write $D'=2b+g$.
\begin{itemize}
\item[(1)] Divisor $D'$ is equal to $3a$. We fire the set $\{a\}$ and obtain the new divisor $a+b+c$.
\item[(2)] Divisor $D'$ is equal to $a+b+c$. We fire the set $\{a,b,c,e,f,g\}$ and obtain the new divisor $a+c+d$.  
\item[(3)] Divisor $D'$ is equal to $a+b+c$. We fire the set $\{a,c\}$ and obtain the new divisor $2b+g$.
\item[(4)] Divisor $D'$ is equal to $2b+g$. We fire the set $\{a,b,c,d,g\}$ and obtain the new divisor $e+2f$.
\end{itemize}

\section{Other notions of gonality}
\label{section:othernotions}

\subsection{Stable divisorial gonality}
The {\em stable divisorial gonality} of a graph $G$ is the minimum of $\dgon(H)$ over all subdivisions $H$ of $G$ (i.e., graphs $H$ that can be obtained by subdividing zero or more edges of $G$).
The bound for divisorial gonality can easily be transferred to one for stable divisorial gonality.
If $G$ is simple, then the treewidth of $G$ equals the treewidth of any of its subdivisions (this is well known).
If $G$ is not simple, then either the treewidth of $G$ equals the treewidth of all its subdivisions, or $G$ is obtained by
adding parallel edges to a forest (i.e., the treewidth of $G$ equals 1), and we subdivide at least one of these parallel edges (thus creating a graph with a cycle; the treewidth will be equal to 2 in this case.) In the latter case, the (stable) divisorial gonality will be at least two. Thus, we have the following easy
corollary.

\begin{corollary}
The treewidth of a graph $G$ is at most the stable divisorial gonality of $G$.
\end{corollary}

Standard treewidth techniques allow us to transform a tree decomposition of a subdivision of $G$ to a tree decomposition of $G$ of the same width. (For each subdivided edge $\{v,w\}$ replace each occurrence of a vertex representing a subdivision of this edge by $v$ in each bag.)

\subsection{Stable gonality}
Related to (stable) divisiorial gonality is the notion of \emph{stable gonality}; see \cite{CKK}.
This notion is defined using finite harmonic morphisms to trees.

Let $G$ and $H$ be undirected nonempty graphs. We allow $G$ and $H$ to have parallel edges but not loops. A graph homomorphism from $G$ to $H$ is a map $f:V(G)\cup E(G)\to V(H)\cup E(H)$ that maps vertices to vertices, edges to edges, and preserves incidences of vertices and edges: 
\begin{itemize}
\item $f(V(G))\subseteq V(H)$,
\item if $e$ is an edge between vertices $u$ and $v$, then $f(e)$ is an edge between $f(u)$ and $f(v)$.
\end{itemize}
A \emph{finite morphism} from $G$ to $H$ (notation: $f:G\to H$) is graph homomorphism $f$ from $G$ to $H$ together with an \emph{index function} $r_f: E(G) \to \ZZ_{>0}$.

A finite morphism $f:G\to H$ with index function $r_f$ is \emph{harmonic}, if for every vertex
$v\in V(G)$, there is a constant $m_f(v)$, such that for each edge $e \in E(H)$ incident to $f(v)$, we have
\[\sum_{e' \,\text{incident to } v; f(e')=e} r_f(e') = m_f(v)\]
If $H$ is connected and $|E(G)| \geq 1$, then there is a positive integer $\deg(f)$, the \emph{degree} of $f$, such that for all vertices $w\in V(H)$ and edges
$e\in E(H)$, we have
\[ \deg(f)= \sum_{v\in V(G); f(v)=w} m_f(v) = \sum_{e' \in E(G); f(e')=e} r_f(e'); \]
see~\cite[Lemma 2.12]{Urakawa2000} and \cite[Lemma 2.3]{baker2009harmonic}.
In particular, $f$ is surjective in this case.

A \emph{refinement} of a graph $G$ is a graph $G'$ that can be obtained from $G$ by zero or more of the following operations: subdivide an edge; add a leaf (i.e., add one new vertex and an edge from that vertex to an existing vertex).

The \emph{stable gonality} of a connected non-empty graph $G$ is the minimum degree of a finite harmonic morphism of a refinement of $G$ to a tree.

\begin{lemma}
Let $G$ be an undirected connected graph without loops and at least one edge. Given a tree $T$ and a finite harmonic morphism $f:G\to T$ of degree $k$, a tree decomposition of $G$ of width at most $k$ can be constructed in $O(k^2|V(G)|)$ time.
\label{lemma:morphismtodecomposition}
\end{lemma}

Before proving the lemma, we make some simple observations. Recall that indices $r_f(e)$ are positive integers. We thus have for each edge $e\in E(T)$:
\[ | \{e'\in E(G) \mid f(e')=e \} | \leq \sum_{e' \in E(G); f(e')=e} r_f(e') = \deg(f). \]

Since $G$ is connected and has at least one edge, it follows that $m_f(v)\geq 1$ for every $v\in V(G)$. Hence, for each vertex $i\in V(T)$: 
\[ | \{v \in V(G) \mid f(v)=i \}| \leq \sum_{v\in V(G); f(v)=i} m_f(v) = \deg(f). \]

\begin{proof}[of Lemma~\ref{lemma:morphismtodecomposition}]
We build a tree decomposition of $G$ in the following way. 
For each edge $e \in E(T)$, we have that $| \{ e'\in E(G) \mid f(e')=e \} | \leq k$. Call this
number $\ell(e)$. We subdivide $e$ precisely $\ell(e)$ times; that is, we add $\ell(e)$ new vertices on this edge. Let $T'$ be the tree that is obtained in this way.

To the nodes $i$ of $T'$, we associate sets $X_i$ in the following way. If $i$ is a node of $T$ (i.e., not a node resulting from the subdivisions), then $X_i = f^{-1}(i)$, i.e., all vertices mapped by
the morphism to $i$. By the observation above, we have that $|X_i|\leq \deg(f) = k$.

Consider an edge $\{i,j\}$ in $T$.  Write $k' = \ell(\{i,j\})$.
Recall that there are $k' \leq k$ edges of $G$ that are mapped to $\{i,j\}$. 
Suppose these are $e_1 = \{v_1,w_1\}, \ldots, e_{k'} = \{v_{k'},w_{k'}\}$
with $f(v_1)=f(v_2)=\cdots =f(v_{k'})= i$ and $f(w_1)=f(w_2)=\cdots =f(w_{k'})= j$.
Let $i_1, i_2, \ldots, i_{k'}$ be the subdivision nodes of the edge $\{i,j\}$, with $i_1$ incident to $i$
and $i_{k'}$ incident to $j$. 
Set $X_{i_r} = \{v_s \mid r\leq s \leq k'\} \cup \{w_t \mid 1 \leq t \leq r\}$ for $r \in \{1, \ldots, k'\}$.
The construction is illustrated in Figure~\ref{fig:gonalitytreewidth}. We claim that this gives a tree decomposition of $G$ of width at most $k$.

\begin{figure}[htb]
    \centering
    \includegraphics[width=\textwidth]{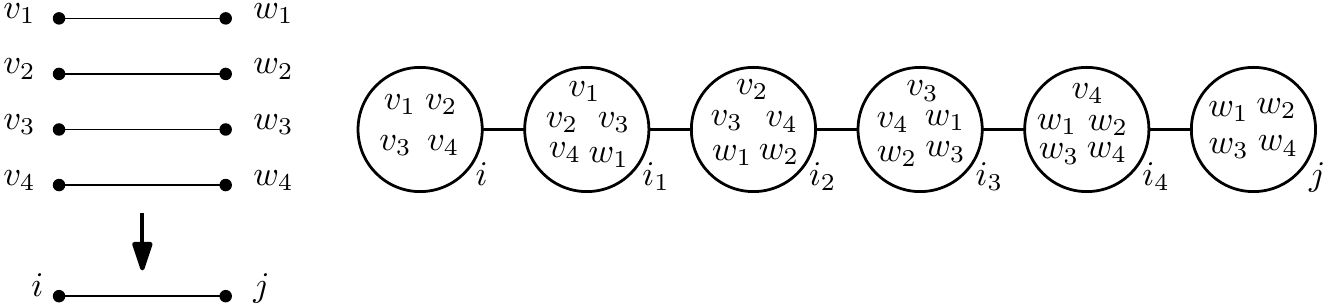}
    \caption{Example of a step in the proof of Lemma~\ref{lemma:morphismtodecomposition}. Here $k'=4$. Left: four edges are mapped to the edge $\{i,j\}$ by the finite harmonic morphism. Right: the corresponding bags in the tree decomposition.}
    \label{fig:gonalitytreewidth}
\end{figure}

 For all edges $\{v,w\}\in E(G)$, we have $\{f(v),f(w)\}\in E(T)$. Suppose without loss of generality that $f(v)$ has the role of $i$, $f(v)$ the role of $j$, $v=v_r$ and $w=w_r$ in the construction above. Then $v, w \in X_{i_r}$.

Finally, for all $v\in V$, the sets $X_i$ to which $v$ belongs are the following: $v$ is in $X_{f(v)}$,
and for each edge incident to $f(v)\in T$, $v$ is in zero or more successive bags of subdivision nodes
of this edge, with the first one (if existing), incident to $f(v)$. Thus, the bags to which $v$ belongs
form a connected subtree.

The first condition of tree decompositions follows from the second and the fact that $G$ is connected.
So, indeed $T'$ with bags as defined above gives a tree decomposition of $G$.

Finally, note that each set $X_i$ is of size at most $k+1$: vertices in $T$ have a bag of size $k$ and
subdivision vertices have a bag of size $k'+1 \leq k+1$. So, we have a tree decomposition of $G$
of width at most $k$.

It is straightforward to see that the construction in the proof can be carried out in $O(k^2|V(G)|)$ time. (Use that $|V(T)| \leq |V(G)|$, since $f$ is surjective.)
\qed
\end{proof}

\begin{theorem}
Let $G$ be an undirected connected graph without loops. Suppose that $G$ has stable gonality $k$. Then $G$ has treewidth at most $k$. Given a refinement $G'$ of $G$ and a finite harmonic morphism $f:G'\to T$ of degree $k$, a tree decomposition of $G$ of width at most $k$ can be constructed in $O(k^2|V(G')|)$ time.
\end{theorem}
\begin{proof}
The degenerate case that $G$ has no edges must be handled separately; here we have that the treewidth of $G$ is $0$, which is equal to its stable gonality.

Suppose $G$ has at least one edge. By Lemma~\ref{lemma:morphismtodecomposition}, we obtain a tree-decomposition of $G'$ of width $k$ in $O(k^2|V(G')|)$ time. Standard treewidth techniques allow us to transform a tree decomposition of a refinement of $G$ to a tree decomposition of $G$ of the same or smaller width. Added leaves can just be removed from all bags where they occur. For each subdivided edge $\{v,w\}$, replace each occurrence of a vertex representing a subdivision of this edge by $v$ in each bag.
\qed
\end{proof}

\subsection*{Acknowledgements}
We thank Gunther Cornelissen, Bart Jansen, Erik Jan van Leeuwen, Marieke van der Wegen, and Tom van der Zanden for helpful discussions.

\bibliographystyle{abbrv}
\bibliography{twgon}

\end{document}